\documentclass[11pt]{article}
\usepackage{natbib}
\usepackage[english]{babel}
\usepackage[latin1]{inputenc}
\usepackage{lmodern}
\usepackage[T1]{fontenc}
\usepackage{amssymb,amsmath,amsthm,latexsym,amsfonts,amscd,dsfont,enumerate,bbm,mathabx}
\usepackage{a4wide,graphicx,tikz}
\usepackage{listings}
\usepackage{booktabs}
\usepackage{mathtools}
\usepackage[gen]{eurosym}
\usepackage{algpseudocode}
\usepackage{soul}

\usepackage{pgfplots}
\pgfplotsset{compat=newest}

\definecolor{mylightyellow}{rgb}{1,1,.8}
\definecolor{mylightgreen}{rgb}{.8,1,.8}
\definecolor{mydarkred}{RGB}{178,34,34}
\definecolor{mydarkgreen}{RGB}{34,139,34}
\definecolor{mydarkblue}{RGB}{72,61,139}
\definecolor{mydarkyellow}{RGB}{218,165,32}

\definecolor{myblueA}{RGB}{52,41,39}
\definecolor{myblueB}{RGB}{92,81,109}
\definecolor{myblueC}{RGB}{132,121,179}
\definecolor{myblueD}{RGB}{172,161,249}

\usetikzlibrary{arrows,patterns,shapes,shadows,fit,backgrounds,positioning,decorations.pathmorphing}
\tikzstyle{nameusd} = [circle, draw, top color=white, bottom color=mydarkblue!50, draw=mydarkblue!75!black!100, drop shadow, minimum height=4em]
\tikzstyle{nameeur} = [circle, draw, top color=white, bottom color=mydarkred!50, draw=mydarkred!75!black!100, drop shadow, minimum height=4em]
\tikzstyle{investor} = [rectangle, draw, top color=white, bottom color=mydarkyellow!50, draw=mydarkyellow!75!black!100, drop shadow, rounded corners, minimum height=3em, text width=4em, text centered]
\tikzstyle{collusd} = [rectangle,fill=mydarkblue!10, inner sep=0.2cm, rounded corners=5mm]
\tikzstyle{colleur} = [rectangle,fill=mydarkred!10, inner sep=0.2cm, rounded corners=5mm]
\tikzstyle{fixto} = [draw, -latex']
\tikzstyle{fixfrom} = [draw, latex'-]
\tikzstyle{floatto} = [draw, snake=coil, segment aspect=0, line before snake=2ex, line after snake=1ex, -latex']
\tikzstyle{floatfrom} = [draw, snake=coil, segment aspect=0, line before snake=2ex, line after snake=1ex, latex'-]

\pgfplotstableread{calibration_coffee.dat}\calibcoffee
\pgfplotstableread{calibration_copper.dat}\calibcopper
\pgfplotstableread{calibration_wti.dat}\calibwti
\pgfplotstableread{calibration_wti_a.dat}\calibwtia
\pgfplotstableread{calibration_wti_lv.dat}\calibwtilv
\pgfplotstableread{cso.dat}\csowti

\usepackage[]{hyperref}
 \hypersetup{
 colorlinks=true,breaklinks=true,
 urlcolor= mydarkblue,linkcolor= mydarkblue,citecolor= mydarkgreen,
 pdfauthor={Nastasi, E.,Pallavicini, A.,Sartorelli, G.}
}

\newtheorem{theorem}{Theorem}[section]

\newtheorem{lemma}[theorem]{Lemma}
\newtheorem{remark}[theorem]{Remark}
\newtheorem{proposition}[theorem]{Proposition}

\newcommand{\Eq}[1]{{\[{#1}\]}}

\newcommand{\Ex}[2]{\mathbb{E}_{#1}\!\left[\,#2\,\right]}

\newcommand{\ExC}[3]{\mathbb{E}_{#1}\!\left[\left.\,#2\,\right|\,#3\,\right]}

\newcommand{\ind}[1]{1_{\{#1\}}}

\newcommand{\var}[3]{{\rm Var}_{#1}^{#2}\!\left[\,#3\,\right]}
\newcommand{\cov}[4]{{\rm Cov}_{#1}^{#2}\!\left[\,#3\,,\,#4\,\right]}
\newcommand{\corr}[4]{{\rm Corr}_{#1}^{#2}\!\left[\,#3\,,\,#4\,\right]}

\usepackage{eurosym}

\DeclareMathOperator*{\argmax}{arg\,max}
\DeclareMathOperator*{\argmin}{arg\,min}

\title{Smile Modelling in Commodity Markets}
\author{
Emanuele Nastasi\thanks{Exprivia, {\tt emanuele.nastasi@exprivia.com}}
\ \ \
Andrea Pallavicini\thanks{Imperial College London and Banca IMI Milan, {\tt a.pallavicini@imperial.ac.uk}}
\ \ \
Giulio Sartorelli\thanks{Banca IMI Milan, {\tt giulio.sartorelli@bancaimi.com}}
}
\date{
\small First Version: October 30, 2017.  This version: \today
}

\begin{document}

\maketitle

\begin{abstract}

We present a stochastic-local volatility model for derivative contracts on commodity futures able to describe forward-curve and smile dynamics with a fast calibration to liquid market quotes. A parsimonious parametrization is introduced to deal with the limited number of options quoted in the market. Cleared commodity markets for futures and options are analyzed to include in the pricing framework specific trading clauses and margining procedures. Numerical examples for calibration and pricing are provided for different commodity products.

\end{abstract}

\bigskip

\noindent {\bf AMS classification codes:} 65C05, 91G20, 91G60.

\newpage
{\small \tableofcontents}
\vfill
{\footnotesize \noindent The opinions here expressed  are solely those of the authors and do not represent in any way those of their employers.}
\newpage

\pagestyle{myheadings} \markboth{}{{\footnotesize  E. Nastasi, A. Pallavicini, G. Sartorelli, Smile Modelling in Commodity Markets}}

\section{Introduction}
\label{sec:introduction}

Futures are the most liquid commodity contracts, followed by option derivatives on futures prices. Plain vanilla options on a single futures price are usually liquid in commodity markets along with calendar spread options, namely spread options on two different futures. However, an increasing number of customized derivative contracts on futures prices are traded over the counter, or they are embedded within structured notes. Since such kind of contracts are usually sensitive to smile effects and include path-dependency, we need a pricing model for futures prices able to describe both curve and smile dynamics. Moreover, for practical purposes, we need a model which can be fast calibrated to futures prices and liquid plain-vanilla options.

Our approach is based on a stochastic-local volatility (SLV) model for futures prices. In the literature some contributions can be found connected to our work. We can cite among them the paper of~\cite{Pilz2011} where the authors introduce a different stochastic volatility model for each futures price along with with a stylized local volatility factor. Stochastic interest rates are considered to size convexity adjustments in futures prices. A different approach is followed in the work of~\cite{Chiminello2015} where a distinct local-volatility model for each futures is introduced. Yet, liquid market quotes for plain-vanilla options on futures prices are quoted only for one expiry date for each futures, so that synthetic market quotes are built from existing ones by means of heuristics to obtain a complete calibration set for each futures. In this model the futures curve dynamics can be modelled with different risk factors for each futures, while smile dynamics is not modelled. A more parsimonious approach is followed by~\cite{Albani2017} where a single local-volatility model for all futures is introduced, but futures curve dynamics and smile dynamics are not modelled. Calibration is performed with quotes observed on different dates to increase the number of market data.

In the present paper we introduce a SLV model able to reproduce futures prices and quoted options on futures, and to recover the price of some exotic option, e.g.\ calendar spread options, which are sensitive to curve and smile dynamics. In particular, we define as a first step a local-volatility one-factor process with affine drift to model future-price marginal densities, and we extend the Dupire equation, see~\cite{Dupire1994} and~\cite{Derman1994}, to allow the implementation of a fast and robust calibration to all quoted plain-vanilla options. The calibration algorithm is constructed as a fixed-point iteration following the accelerated Anderson scheme, see~\cite{Walker2011} and~\cite{Anderson1965}. Then, as second step we introduce a SLV dynamics for each futures price, so that we can model its curve and smile dynamics. In doing so we exploit the Gy\"ongy Lemma, see~\cite{Gyongy1986}, to match the future-price marginal densities calibrated in the first step. The present contribution focuses mostly on the local-volatility projection, while we leave to a future work the detailed analysis of stochastic parameters calibration and the discussion of exotic product pricing.

\medskip

The paper is organized as follows. In Section~\ref{sec:market} we describe the markets quoting commodity futures and options on futures. In particular, we discuss the trading clauses and the margining rules. Then, in Sections~\ref{sec:modelling} and~\ref{sec:calibration} we present the local-volatility part of the model and we describe how to calibrate futures prices and options on futures, while the extension of the model to a full SLV model is completed in Section~\ref{sec:smile}.

\section{Futures and Options on Futures in Commodity Markets}
\label{sec:market}

Liquid plain-vanilla options on commodity markets give the right to enter a futures contract at a given price. The decision to enter the contract can be done at any time up to option maturity, usually within a couple of years since trading trade. Option contracts are either margined as futures or collateralized as equity derivatives. We refer as an example to the documentation provided by ICE and CME markets\footnote{The Intercontinental Exchange (ICE, \url{https://www.intercontinentalexchange.com}) and the Chicaco Mercantile Exchange (CME, \url{http://www.cmegroup.com}) are two of the leading electronic trading platforms for commodities.}.

\subsection{Futures Contracts}

We wish to introduce a single-factor background process whose marginal probability densities could be mapped on futures densities so that we can parsimoniously price all liquid options. We term this process the ``fictitious spot'' price process since in our model its value plays the role of a spot price, even if it is not necessarily linked to an observable quantity. Once futures marginal probability densities are calibrated we can add curve and smile dynamics. Before doing so we need to understand the specifics of futures contracts.

In commodity markets futures contracts may lead to the physical delivery of the goods that has been purchased or sold in advance. After the first notification date the holder of a long position in the futures contract may be notified by the Exchange that the goods will be delivered. In order to avoid a physical delivery, the holder of the long position must roll the futures contract before the first notification date. We show in Figure~\ref{fig:market} some examples of trading and notification dates for futures contracts and options on futures. In detail we have that futures contracts are traded in the market between dates $T^f_0$ and $T^f_1$, but usually rolled before the first notification date $T^n_0$ to avoid delivery between dates $T^d_0$ and $T^d_1$. Notification of delivery may occur up to the last notification date $T^n_1$. Before date $T^n_0$ options are traded between dates $T^o_0$ and $T^o_1$. Futures quoted after the start delivery date $T^d_0$ refer to contracts with a shorter delivery period.

In the following, we assume that at time $t$ the prices $F_t^{\textrm{ID}}$ of a family of futures contracts each identified by a label ID are available in the market. A first example is given by Coffee C Futures on ICE\@. We show the time structure of the ID=''{\tt MAR18}'' futures and option contracts in the upper panel of Figure~\ref{fig:market}, where we plot both the time-line for futures contracts (in red) and for option on futures (in blue). Continuous color bands refers to the trading windows of each contract up to the first notification date. We can see that the first notification and the first delivery dates occurs well before the last trading date. The same is happening for Copper Futures on CME displayed in the mid panel again for {\tt MAR18} futures and option contracts. On the other hand for WTI Crude Oil Futures on CME notification and delivery may occur only after the last trading date, as shown in the third example in the lower panel of the figure where the time structure of {\tt MAR18} futures and option contracts are displayed.

\begin{remark} {\bf (Delivery Periods)}
Some commodity markets, such as natural gas or electricity, quote futures contracts on different set of delivery periods. In natural gas markets a futures contract allows to enter in a daily supply of gas for one month, a quarter or the whole year. Also options are traded on these different futures sets. In this markets we choose to model futures contracts with only one specific delivery period. We can focus on one-month delivery period contracts, and we can still price contracts with longer periods by implementing strategies on shorter contracts. However, in this way we cannot price contracts with shorter period. In these markets we can use option quotes on longer periods to get informations on forward curve dynamics.
\end{remark}

\begin{figure}
\begin{center}
\begin{tikzpicture}

\fill [mydarkred!25] (0.2,0) rectangle (8.0,0.5);
\fill [pattern=north west lines,pattern color=white] (3.5,0) rectangle (8.0,0.5);
\fill [mydarkblue!25] (0.8,-1.5) rectangle (2.5,-1.0);

\draw[dotted,line width=1pt,color=mydarkred!75] (3.5,0) -- (3.5,2) node [right] {{\footnotesize$F_{T_0^n}^{\tt MAR18}$}};
\draw[dotted,line width=1pt,color=mydarkblue!75] (2.5,-1.5) -- (2.5,-2.5) node [right] {{\footnotesize$\left(F_{T_1^o}^{\tt MAR18}-K\right)^{\!+}$}};

\draw[line width=2pt] (0,0) -- (1,0);
\draw[dotted,line width=2pt] (1,0) -- (2,0);
\draw[->,line width=2pt] (2,0) -- (10.5,0);

\draw[line width=2pt] (0,-1.5) -- (1,-1.5);
\draw[dotted,line width=2pt] (1,-1.5) -- (2,-1.5);
\draw[->,line width=2pt] (2,-1.5) -- (3.5,-1.5);

\draw (0.2,0) node [below] {{\footnotesize$\phantom{T^f_0}T^f_0$}} -- (0.2,-0.1) node [right,rotate=90,xshift=5pt,yshift=5pt] {{\footnotesize \scalebox{.7}[1.0]{01 Apr 2015}}};
\draw (8.0,0) node [below] {{\footnotesize$\phantom{T^f_0}T^f_1$}} -- (8.0,-0.1) node [right,rotate=90,xshift=5pt,yshift=5pt] {{\footnotesize \scalebox{.7}[1.0]{19 Mar 2018}}};

\draw (3.5,0) node [below] {{\footnotesize$\phantom{T^f_0}T^n_0$}} -- (3.5,-0.1) node [right,rotate=90,xshift=5pt,yshift=5pt] {{\footnotesize \scalebox{.7}[1.0]{20 Feb 2018}}};
\draw (8.5,0) node [below] {{\footnotesize$\phantom{T^f_0}T^n_1$}} -- (8.5,-0.1) node [right,rotate=90,xshift=5pt,yshift=0pt] {{\footnotesize \scalebox{.7}[1.0]{20 Mar 2018}}};

\draw (6.5,0) node [below] {{\footnotesize$\phantom{T^f_0}T^d_0$}} -- (6.5,-0.1) node [right,rotate=90,xshift=5pt,yshift=0pt] {{\footnotesize \scalebox{.7}[1.0]{01 Mar 2018}}};
\draw (9.5,0) node [below] {{\footnotesize$\phantom{T^f_0}T^d_1$}} -- (9.5,-0.1) node [right,rotate=90,xshift=5pt,yshift=0pt] {{\footnotesize \scalebox{.7}[1.0]{29 Mar 2018}}};

\draw (0.8,0) node [right,rotate=90,xshift=1.9pt,yshift=0pt] {{\footnotesize \scalebox{.7}[1.0]{02 Apr 2015}}} -- (0.8,-1.6) node [below] {{\footnotesize$\phantom{T^o_0}T^o_0$}};
\draw (2.5,0) node [right,rotate=90,xshift=1.9pt,yshift=0pt] {{\footnotesize \scalebox{.7}[1.0]{09 Feb 2018}}} -- (2.5,-1.6) node [below] {{\footnotesize$\phantom{T^o_0}\,T^o_1$}};

\end{tikzpicture}
\end{center}

\begin{center}
\begin{tikzpicture}

\fill [mydarkred!25] (0.2,0) rectangle (8.0,0.5);
\fill [pattern=north west lines,pattern color=white] (3.5,0) rectangle (8.0,0.5);
\fill [mydarkblue!25] (0.8,-1.5) rectangle (2.5,-1.0);

\draw[dotted,line width=1pt,color=mydarkred!75] (3.5,0) -- (3.5,2) node [right] {{\footnotesize$F_{T_0^n}^{\tt MAR18}$}};
\draw[dotted,line width=1pt,color=mydarkblue!75] (2.5,-1.5) -- (2.5,-2.5) node [right] {{\footnotesize$\left(F_{T_1^o}^{\tt MAR18}-K\right)^{\!+}$}};

\draw[line width=2pt] (0,0) -- (1,0);
\draw[dotted,line width=2pt] (1,0) -- (2,0);
\draw[->,line width=2pt] (2,0) -- (10.5,0);

\draw[line width=2pt] (0,-1.5) -- (1,-1.5);
\draw[dotted,line width=2pt] (1,-1.5) -- (2,-1.5);
\draw[->,line width=2pt] (2,-1.5) -- (3.5,-1.5);

\draw (0.2,0) node [below] {{\footnotesize$\phantom{T^f_0}T^f_0$}} -- (0.2,-0.1) node [right,rotate=90,xshift=5pt,yshift=5pt] {{\footnotesize \scalebox{.7}[1.0]{28 Mar 2013}}};
\draw (8.0,0) node [below] {{\footnotesize$\phantom{T^f_0}T^f_1$}} -- (8.0,-0.1) node [right,rotate=90,xshift=5pt,yshift=5pt] {{\footnotesize \scalebox{.7}[1.0]{27 Mar 2018}}};

\draw (3.5,0) node [below] {{\footnotesize$\phantom{T^f_0}T^n_0$}} -- (3.5,-0.1) node [right,rotate=90,xshift=5pt,yshift=5pt] {{\footnotesize \scalebox{.7}[1.0]{28 Feb 2018}}};
\draw (8.5,0) node [below] {{\footnotesize$\phantom{T^f_0}T^n_1$}} -- (8.5,-0.1) node [right,rotate=90,xshift=5pt,yshift=0pt] {{\footnotesize \scalebox{.7}[1.0]{28 Mar 2018}}};

\draw (6.5,0) node [below] {{\footnotesize$\phantom{T^f_0}T^d_0$}} -- (6.5,-0.1) node [right,rotate=90,xshift=5pt,yshift=0pt] {{\footnotesize \scalebox{.7}[1.0]{01 Mar 2018}}};
\draw (9.5,0) node [below] {{\footnotesize$\phantom{T^f_0}T^d_1$}} -- (9.5,-0.1) node [right,rotate=90,xshift=5pt,yshift=0pt] {{\footnotesize \scalebox{.7}[1.0]{29 Mar 2018}}};

\draw (0.8,0) node [right,rotate=90,xshift=1.9pt,yshift=0pt] {{\footnotesize \scalebox{.7}[1.0]{27 Apr 2016}}} -- (0.8,-1.6) node [below] {{\footnotesize$\phantom{T^o_0}T^o_0$}};
\draw (2.5,0) node [right,rotate=90,xshift=1.9pt,yshift=0pt] {{\footnotesize \scalebox{.7}[1.0]{22 Feb 2018}}} -- (2.5,-1.6) node [below] {{\footnotesize$\phantom{T^o_0}\,T^o_1$}};

\end{tikzpicture}
\end{center}

\begin{center}
\begin{tikzpicture}

\fill [mydarkred!25] (0.2,0) rectangle (3.5,0.5);
\fill [mydarkblue!25] (0.2,-1.5) rectangle (2.5,-1.0);

\draw[dotted,line width=1pt,color=mydarkred!75] (3.5,0) -- (3.5,2) node [right] {{\footnotesize$F_{T_1^f}^{\tt MAR18}$}};
\draw[dotted,line width=1pt,color=mydarkblue!75] (2.5,-1.5) -- (2.5,-2.5) node [right] {{\footnotesize$\left(F_{T_1^o}^{\tt MAR18}-K\right)^{\!+}$}};

\draw[line width=2pt] (0,0) -- (1,0);
\draw[dotted,line width=2pt] (1,0) -- (2,0);
\draw[->,line width=2pt] (2,0) -- (10.5,0);

\draw[line width=2pt] (0,-1.5) -- (1,-1.5);
\draw[dotted,line width=2pt] (1,-1.5) -- (2,-1.5);
\draw[->,line width=2pt] (2,-1.5) -- (3.5,-1.5);

\draw (0.2,0) node [below] {{\footnotesize$\phantom{T^f_0}T^f_0$}} -- (0.2,-0.1) node [right,rotate=90,xshift=5pt,yshift=5pt] {{\footnotesize \scalebox{.7}[1.0]{19 Nov 2012}}};
\draw (3.5,0) node [below] {{\footnotesize$\phantom{T^f_0}T^f_1$}} -- (3.5,-0.1) node [right,rotate=90,xshift=5pt,yshift=5pt] {{\footnotesize \scalebox{.7}[1.0]{20 Feb 2018}}};

\draw (4.5,0) node [below] {{\footnotesize$\phantom{T^f_0 \equiv T^f_1}T^n_0 \equiv T^n_1$}} -- (4.5,-0.1) node [right,rotate=90,xshift=5pt,yshift=0pt] {{\footnotesize \scalebox{.7}[1.0]{22 Feb 2018}}};

\draw (6.5,0) node [below] {{\footnotesize$\phantom{T^f_0}T^d_0$}} -- (6.5,-0.1) node [right,rotate=90,xshift=5pt,yshift=0pt] {{\footnotesize \scalebox{.7}[1.0]{01 Mar 2018}}};
\draw (9.5,0) node [below] {{\footnotesize$\phantom{T^f_0}T^d_1$}} -- (9.5,-0.1) node [right,rotate=90,xshift=5pt,yshift=0pt] {{\footnotesize \scalebox{.7}[1.0]{31 Mar 2018}}};

\draw (0.2,0) -- (0.2,-1.6) node [below] {{\footnotesize$\phantom{T^o_0}T^o_0$}};
\draw (2.5,0) node [right,rotate=90,xshift=1.9pt,yshift=5pt] {{\footnotesize \scalebox{.7}[1.0]{14 Feb 2018}}} -- (2.5,-1.6) node [below] {{\footnotesize$\phantom{T^o_0}\,T^o_1$}};

\end{tikzpicture}
\end{center}
\caption{Notification, trading and delivery dates for March 2018 futures along with trading dates for corresponding options. Upper panel Coffee Arabica on ICE, mid panel Copper on CME, lower panel WTI Crude Oil on CME.}
\label{fig:market}
\end{figure}
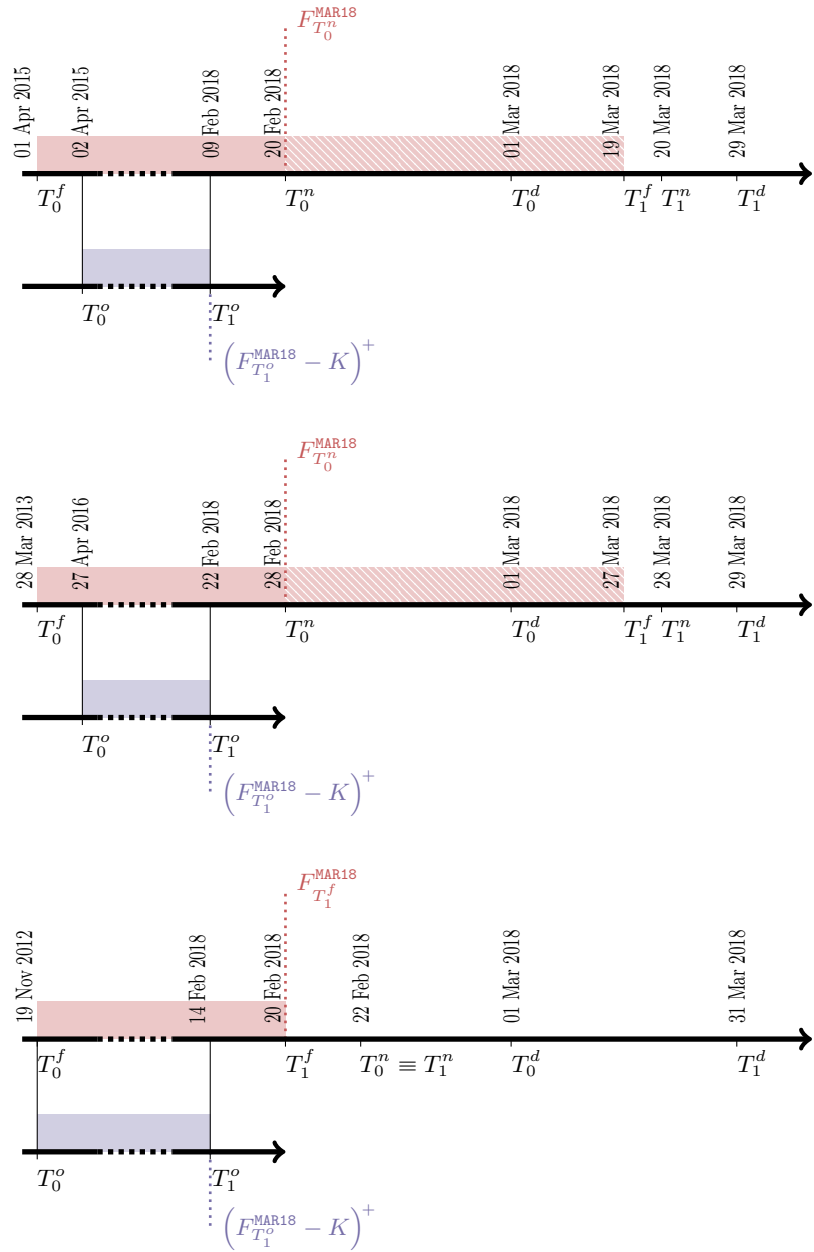

\subsection{Margining and Collateral Procedures}

Futures exchanges determine and set futures margin rates. Traders are required to post an initial margin at contract inception, then the margin in maintained to match the mark-to-market variations of the contract. The party posting cash or assets to fulfill the margin agreement is not remunerated for funding costs. The procedure is similar to what happens for collateralized contracts, but in this latter case the collateral taker pays an interest rate to other party, usually the over-night rate. Similar arguments hold also for options on futures. The presence or absence of the remuneration rate is crucial to set up the right pricing formula for these contracts.

Pricing in presence of margining or collateral procedures is widely discussed in the literature. Here, we refer to~\cite{BrigoMoriniPallavicini2012} and references therein. In particular, we assume the so-called perfect collateralization regime, namely the case when the margining procedure is able to remove any losses on the default event of the investor or the counterparty. We have the following proposition.

\begin{proposition} {\bf (Perfect collateralization)}\\
We assume that the market risks are described by a vector of It\^o processes. We consider a security with price process $V_t$ and cumulative dividend processes $\pi_t$ expressing the contractual coupons $\phi_i$ paid on date $T_i$ defined as
\begin{equation}
\label{eq:dumdiv}
\pi_t := \sum_{i=1}^N \phi_i \ind{T_i \le t}
\end{equation}
If the collateral process $C_t$ is defined so that $C_t = V_t$ for any time $t$ up to security maturity $T$, then we can write the security price process as
\begin{equation}
\label{eq:perfect}
V_t = \,\Ex{t}{\int_t^T D(t,u;e) \,d\pi_u}
\;,\quad
D(t,T;e) := \exp\left\{ - \int_t^T e_u \,du \right\}
\end{equation}%
where the expectation is taken under the risk-neutral measure, and $e_t$ is the collateral accrual-rate process.
\label{cor:perfect}
\end{proposition}
\begin{proof}
The proposition can be derived from Corollary 2.1 in~\cite{Moreni2017} by assuming that all the payment currencies are the same.
\end{proof}

We are now ready to write the pricing equations for options on futures. We start by describing plain-vanilla options without early-exercise features (European options).  They can be traded either under a margin or a collateral agreement, usually termed by futures exchanges respectively future-style and equity-style margining. Under future-style margining the option premium is paid on expiry date, mark-to-market variations are exchanged on a daily basis along with initial margins without any remuneration, so that we can write the option price according to Proposition~\ref{cor:perfect} as given by
\begin{equation}
C^{\rm marg}_t(T^o_1,\textrm{ID},K) := \Ex{t}{\left(F_{T_1^o}^{\textrm{ID}} - K\right)^{\!+}}.
\label{eq:optmarg}
\end{equation}%
On the other hand, under equity-style margining the option premium is paid on trade date, the option seller must post a guarantee for the premium (collateral accruing at rate $e$) which is adjusted on a daily basis along with initial margins, so that we can write the option price according to Proposition~\ref{cor:perfect} as given by
\begin{equation}
C^{\rm coll}_t(T^o_1,\textrm{ID},K) := \Ex{t}{\left(F_{T_1^o}^{\textrm{ID}}-K\right)^{\!+} D(t,T^p_1;e)},
\label{eq:optcoll}
\end{equation}%
where $T^p_1$ is the payment date and usual settlement lags apply. In case of early-exercise we should adjust our pricing formula, as done for instance in~\cite{Albani2017}. However, for non-dividend paying assets American option prices may be different from corresponding European prices only because of interest-rates. In the case of future-style margining we can see from Equation~\eqref{eq:optmarg} that option prices do not depend on interest rates, since the discount factor is missing and the underlying asset is a futures price. Thus, we can conclude that it is never convenient to exercise the option before maturity, namely the prices of European and American options are the same when they are traded under future-style margining agreements, and so adjustments are not required. On the other hand, options traded with equity-style margining require adjustments, although in low interest-rate regime and for short maturities they should be small. If adjustments are required we can easily include them by following~\cite{DeMarco2016} which extends the Dupire equation to American options.

\section{Modelling Futures Prices}
\label{sec:modelling}

We follow a step-by-step procedure to implement our modelling framework. We start by introducing a local-volatility linear model for a background process we name the ``fictitious spot price''. Linear models are characterized by an affine drift term, see for instance~\cite{Ackerer2016a}. In commodity energy markets a constant-volatility version of linear models is studied in~\cite{Swishchuk2008}. Such models allow to calculate futures prices in closed form. Moreover, in Section~\ref{sec:calibration} we will show that option prices may be calculated by an extended version of the Dupire equation, which we can employ to calibrate the local volatility function.

\subsection{Dealing with the Delivery Procedure}

In the mathematical representation of equity futures markets the underlying of the futures contract is a spot contract which is usually fairly liquid. Standard assumptions on absence of arbitrage allow to obtain that an equivalent probability measure exists such that the futures price process is a martingale, and its value coincides with the spot price process value at the futures expiry date.

In commodity markets the futures contracts give the right to exchange an underlying physical commodity for an amount of money at expiry, but the delivery procedure is non-trivial due to the physical nature of the underlying, so that it cannot be summarized as ``the delivery of an amount of commodity at one date''. The delivery procedure can take several days according to rules allowing for optional choices by the counterparties, leading to a ``spot'' contract whose valuation can be difficult to perform. This prevents us from using the same arguments used in the equity case, nevertheless we can overcome such problems if the final purpose of the model is only valuing derivative contracts on futures prices, by taking into account the availability of the futures themselves as hedging instruments. In this case we only need an equivalent measure such that the futures price process is a martingale up to any date before the delivery procedure may begin, namely before the first notification date $T_0^n$. Indeed, we are quite close to reality if we say that there is no optionality granted by futures contracts, beyond that of buying and selling the futures itself, so absence of arbitrage implies that there is an equivalent measure such that at each day up to the first notification date the futures price is the expected value of the futures price at the next day, if this comes before the end of trading at $T_1^f$. After that, various optionalities may ensue, in many cases making it dubious that the futures price process is still a martingale with respect to any equivalent measure. Hence, for any time $t\leq T^{\textrm{last}}$ with $T^{\textrm{last}} = \min\{ T_0^n, T_1^f\}$ we can write
\[
F_t^{\textrm{ID}}= \Ex{t}{ F_{T^{\textrm{last}}}^{\textrm{ID}}}
\]%
and we can assume the existence of a single continuous-time process $S_t$, which we term \emph{fictitious spot process}, whose value coincides with futures contract prices on the $T^{\textrm{last}}$ date of each futures contract, so that
\[
F_t^{\textrm{ID}} = \Ex{t}{S_{T^{\textrm{last}}}},
\quad
t \leq T^{\textrm{last}}
\]%
We do not require that $S_t$ is traded on the market. If the market quotes a spot price, it can be different from our definition, for instance due to particular delivery conditions.

Moreover, even though only futures for a finite number of expiries are observable and tradable on the market, we may postulate the existence of a whole curve $T \mapsto F(T)$, parameterized by the date $T$ at which the futures price coincides with the fictitious spot price:
\[
F_T(T) = S_T
\]%
for traded futures, in this setting we have:
\[
F_{T^{\textrm{last}}}(T^{\textrm{last}}) = F_{T^{\textrm{last}}}^{\textrm{ID}} = S_{T^{\textrm{last}}}
\]%

\subsection{Modelling the Fictitious Spot Price}

We wish to select a model for the fictitious spot price $S_t$ which allows us to easily calculate the forward prices, but, at the same time, it is flexible enough to reproduce all option quotes. Thus, we choose a local-volatility linear model for the fictitious spot price $S_t$.
\begin{equation}
dS_t = \left( \alpha(t) + \beta(t) S_t \right) dt + \eta_S(t,S_t) S_t \,dW_t
\;,\quad
S_0 = {\bar S}
\label{eq:spot}
\end{equation}%
where $W$ is a standard Brownian motion under risk-neutral measure, $\alpha$ is a positive function of time, $\beta$ is a function of time, $\eta_S$ is Lipschitz, bounded and positive, ${\bar S}$ is the (positive) initial value of the fictitious spot price, which is actually irrelevant in the modeling of futures prices. With these assumptions the previous SDE has a unique positive solution for any time $t>0$ and finite moments of all orders.

Our strategy is to calibrate $\beta$ to futures prices and $\eta_S$ to plain-vanilla option prices, while $\alpha$ can be derived from calendar spread options, if their quotes are liquid in the market, or by means of heuristic arguments. The specific form of the dynamics of the fictitious spot price is selected to perform in an effective way these tasks.

In the following we find easier to work with a normalized version of the spot price, so we define
\begin{equation}
s_t := \frac{S_t}{F_0(t)}
\label{eq:fictitious}
\end{equation}%
whose dynamics is
\begin{equation}
ds_t = \left( a(t) + (\beta(t) - \partial_t \ln F_0(t)) s_t \right) dt + \eta(t,s_t) s_t \,dW_t
\;,\quad
s_0 = 1
\label{eq:normspot}
\end{equation}%
where the coefficients are defined as
\Eq{
a(t) := \frac{\alpha(t)}{F_0(t)} > 0
\;,\quad
\eta(t,s_t) := \eta_S(t,s_tF_0(t))
}%
and time $0$ represents when we perform the calibration procedure.

\begin{remark} {\bf (Speed of Mean Reversion)}
We introduce a speed of mean reversion $a$ since we notice, by practical investigation, that a fictitious spot process with $a=0$ may fail to calibrate option market data. Indeed, the existence of a fictitious spot price is only a model assumption, and it is not driven by non-arbitrage pricinciples as, for instance, in the Equity market case. In the following, when we discuss the calibration procedure, we have to deal with two unknown parameters: the local-volatility function $\eta$ and the speed of mean reversion $a$. In the following we will calibrate the local volatility function to plain-vanilla options, while we will discuss at the end of this section how to fix the mean-reversion speed.
\label{remark:meanrev}
\end{remark}

\subsection{Automatic Calibration of Futures Prices}

We proceed with the calibration of futures prices. They can be automatically recovered if we properly define the function $\beta$ in the spot dynamics.

As a first step we re-write Equation~\eqref{eq:normspot} in integral form.
\Eq{
s_T = 1 + \int_0^T \left( a(u) + (\beta(u) - \partial_u \ln F_0(u)) s_u \right) du + \int_0^T \eta(u,s_u) s_u \,dW_u
}%
We can take the expectation under risk-neutral measure conditioning at any time $t\in[0,T)$ to obtain
\Eq{
\Ex{t}{s_T} = 1 + \int_0^T \left( a(u) + (\beta(u) - \partial_u \ln F_0(u)) \Ex{t}{s_u} \right) du
}%
Then, if we use the definition of futures prices in term of the fictitious spot price we get
\Eq{
\frac{F_t(T)}{F_0(T)} = 1 + \int_0^T \left( a(u) + (\beta(u) - \partial_u \ln F_0(u)) \frac{F_t(u)}{F_0(u)} \right) du
}%
Hence, we can take the derivative w.r.t.\ maturity time $T$ to write a first order ODE for futures prices
\begin{equation}
\partial_T F_t(T) = a(T) F_0(T) + \beta(T) F_t(T)
\;,\quad
F_t(t) = s_t F_0(t)
\label{eq:fwdode}
\end{equation}%

In particular, for $t=0$ we can re-arrange Equation~\eqref{eq:fwdode} to obtain an explicit expression for $\beta(t)$ that we can substitute in the spot dynamics, namely we get
\Eq{
\beta(T) = \partial_T \ln F_0(T) - a(T)
}%
leading to the following dynamics for normalized spot prices
\begin{equation}
ds_t = a(t) (1 - s_t) \,dt + \eta(t,s_t) s_t \,dW_t
\;,\quad
s_0 = 1
\label{eq:normspot_calib}
\end{equation}%
In this way we obtain that, if the fictitious spot price follows the above dynamics, the futures prices are exactly recovered by our model, since from Equation~\eqref{eq:fictitious} we have
\Eq{
S_t = F_0(t) s_t \;\Longrightarrow\; \Ex{0}{S_t} = F_0(t) \Ex{0}{s_t} = F_0(t)
}%
where the expectation on the right hand side can be easily computed given the dynamics of $s_t$ described in Equation~\eqref{eq:normspot_calib}.

We conclude this section by explicitly solving the first-order ODE described by Equation~\eqref{eq:fwdode}, so that we can write the dynamics followed by the futures prices. We will use this result in the following sections. We obtain
\begin{equation}
F_t(T) = F_0(T) \left( 1 - (1-s_t) e^{-\int_t^T a(u) \,du} \right)
\label{eq:fwdspot}
\end{equation}%
Then, by differentiating w.r.t.\ to time $t$ we get
\begin{equation}
dF_t(T) = \eta_F(t,T,F_t(T)) \,dW_t
\label{eq:fwd}
\end{equation}%
where the local volatility of futures prices can be defined by means of a proper remapping of the local volatility of the spot price. Indeed, we get
\begin{equation}
\eta_F(t,T,K) := \left( K - F_0(T) \left( 1 - e^{-\int_t^T a(u) \,du} \right) \right) \eta(t,k_F(t,T,K))
\label{eq:etaF}
\end{equation}%
while the the effective strike $k_F$ can be defined as
\begin{equation}
k_F(t,T,K) := 1 - e^{\int_t^T a(u) \,du} \left(1 - \frac{K}{F_0(T)}\right)
\label{eq:effK}
\end{equation}%

\section{Calibration of Futures Option Smile}
\label{sec:calibration}

The last step of the calibration procedure for the fictitious spot price consists in finding the local volatility function $\eta$ able to recover the prices of options on futures. We wish to achieve a fast calibration procedure, so that we search for an extension of the Dupire equation for our model, which allows us to price all plain-vanilla options by a single PDE evaluation.

\subsection{Extended Dupire Equation}

Here, we focus on European options, since the discussion of Section~\ref{sec:market} on early-exercise options showed that future-style American options have the same price as European options, while equity-style options can be calculated as in~\cite{DeMarco2016} by a proper modification of the Dupire equation, which is compatible with our extension.

\begin{proposition} {\bf (Extended Dupire Equation)}\\
We assume that the normalized spot price $s_t$ follows the dynamics
\Eq{
ds_t = a(t) (1 - s_t) \,dt + \eta(t,s_t) s_t \,dW_t
\;,\quad
s_0 = 1
}%
where the mean-reversion speed $a$ is a positive function of time, and the local-volatility function $\eta$ is Lipschitz, bounded and positive. Then, the normalized call price
\Eq{
c(t,k) := \Ex{0}{\left(s_t - k\right)^{\!+}}
}%
satisfies the following parabolic PDE
\begin{equation}
\partial_t c(t,k) = \left( - a(t) - a(t) (1-k) \,\partial_k + \frac{1}{2} k^2 \eta^2(t,k) \,\partial^2_k \right) c(t,k)
\label{eq:dupire}
\end{equation}%
with boundary conditions
\Eq{
c(t,0) = 1
\;,\quad
c(t,\infty) = 0
\;,\quad
c(0,k) = (1-k)^{\!+}
}%
\label{prop:dupire}
\end{proposition}
\begin{proof}
We can proceed by applying the Meyer-$\!$Tanaka Formula to derive the dynamics of the call price, see for instance~\cite{protter}, and we write
\Eq{
c(t,k) = c(0,k) + \Ex{0}{ \int_0^t  \ind{s_t>k} ds_u } + \frac{1}{2} \,\Ex{0}{ L^s_t(k) }
}%
where the local time $L^s_t(k)$ for the process $s_t$ at level $k$ is defined as
\Eq{
L^s_t(k) := \lim_{\epsilon\rightarrow 0^+} \frac{1}{2\epsilon} \int_0^t \ind{k-\epsilon \le s_u \le k+\epsilon} \,d\langle s \rangle_u
}%
Then, following the results of~\cite{Bentata2015}, we can expand the expectations in term of integrals over the risk-neutral price density $p_{s_u}(x)$, so that we can write
\Eq{
c(t,k) = c(0,k) + \int_0^t a(u) \int_k^\infty  p_{s_u}(x) (1-x) dx\,du + \frac{1}{2} k^2  \int_0^t  p_{s_u}(k) \eta^2(u,k) \,du
}%
On the other hand, we have $p_{s_u}(k) = \partial^2_k c(u,k)$. Thus, by direct substitution and by differentiating w.r.t.\ time $t$ we get
\Eq{
\partial_t c(t,k) = a(t) \int_k^\infty  (1-x) \,\partial^2_k c(t,x) \,dx + \frac{1}{2} k^2 \eta^2(t,k) \,\partial^2_k c(t,k)
}%
We can simplify the equation by integrating by parts twice the first term on the righ-hand side. Indeed, we have
\Eq{
\int_k^\infty  (1-x) \,\partial^2_k c(t,x) \,dx = \left. \left( c(t,x) + (1-x) \partial_k c(t,x) \right) \right|_k^\infty = - c(t,k) - (1-k) \partial_k c(t,k)
}%
where the last equality holds since the finiteness of all moments implies that call prices decrease at large strike faster than a polynomial, see~\cite{Lee2004}. If we assemble the results we get the proposition.

\end{proof}

A similar result can be found in~\cite{Deelstra2013} for long-dates FX options, but here we can go further on by exploiting the affine form of the drift to remove the integral term, which allows to solve Equation~\eqref{eq:dupire} by means of the implicit PDE discretization method in a fast and efficient way as usually done for the standard Dupire equation.

Once the normalized call prices are calculated, we can also derive the prices of options on futures. Indeed, we get by direct substitution the price of future-style options
\begin{equation}
C^{\rm marg}_0(t,T,K) = F_0(T) e^{-\int_t^T a(u) \,du} c(t,k_F(t,T,K))
\label{eq:fwdopt_marg}
\end{equation}%
and of equity-style options
\begin{equation}
C^{\rm coll}_0(t,T,K) = P_0(T_p;e) F_0(T) e^{-\int_t^T a(u) \,du} c(t,k_F(t,T,K))
\label{eq:fwdopt_coll}
\end{equation}%
where the effective strike $k_F$ is defined in Equation~\eqref{eq:effK}.

\subsection{Calibration of the Local Volatility Function}

There is a huge literature on the calibration of local volatility models by using the Dupire equation, see for instance~\cite{gatheral2006volatility} and references therein. Here, we stress the need of a fast and robust calibration procedure to allow practical applications on trading desks. In particular, we may solve Equation~\eqref{eq:dupire} with two different approaches: (i) given the prices of options on futures from the market, we can use Equations~\eqref{eq:fwdopt_marg} and~\eqref{eq:fwdopt_coll} to obtain the corresponding normalized call prices, then we can plug these prices in Equation~\eqref{eq:dupire} and solve for the local volatility function, or (ii) we can solve a best-fit problem to find the local-volatility function that, plugged in the PDE~\eqref{eq:dupire}, gives the proper normalized call prices to be put in Equations~\eqref{eq:fwdopt_marg} and~\eqref{eq:fwdopt_coll} to get the market prices.

The first approach seems to be much faster than the second one since it does not require a best-fit procedure based on a PDE solver. Yet, it needs an arbitrage-free interpolation scheme to calculate partial derivatives from market prices in Equation~\eqref{eq:dupire}. A precise and robust definition of such scheme may result in slow algorithms and it may develop instabilities due to the small value of the partial derivatives in very out- or in-the-money strike regions. Here, we follow the second approach and we must face the problem of limiting the number of times we re-evaluate the PDE within the best-fit procedure.

We consider a set of options on futures prices quoted in term of Black-Scholes volatilities which we identify as
\Eq{
\{\sigma_F^{\rm mkt}(t_i,\textrm{ID}_i,K_i^j) : i\in[1,N], j\in[1,M_i]\}
}%
We notice that strike prices may depend on maturities since the market usually quotes options in term of Black-Scholes $\Delta$ leading to a different set of strikes for each maturity.  Using~\eqref{eq:fwdopt_marg} or~\eqref{eq:fwdopt_coll} we can transform market quotes of options on futures into volatilities of European options on the normalized spot price
\Eq{
\{\sigma^{\rm mkt}(t_i,k_F(t_i,T_i^{\text{last}},K_i^j)) : i\in[1,N], j\in[1,M_i]\}
}%
where we define
\Eq{
\sigma^{\rm mkt}_{ij} := \sigma^{\rm mkt}(t_i,k_F(t_i,T_i^{\text{last}},K_i^j))
}%
Adjustments for American early-exercise features may be included when options are equity-style collateralized. Model calibration to options on futures prices can be achieved by specifying a particular non-parametric form for the local volatility
\Eq{
\eta(t,k) := \zeta(t,k;\{\eta(t_i,k_F(t_i,T_i^{\text{last}},K_i^j))\})
\;,\quad
\eta_{ij} := \eta(t_i,k_F(t_i,T_i^{\text{last}},K_i^j))
}%
where $\zeta$ is a function used to interpolate and extrapolate the local-volatility function. We select a cubic monotone spline for interpolation, and constant function for extrapolation.

The best-fit problem has as free parameters the nodes of the local-volatility function, so that we have as many parameters as market volatilities, usually up to 100 or more parameters. Calibration procedures based on gradient-based optimization algorithms may suffer of poor performances, since the Jacobian of the objective function is always calculated in all directions. If we have a good knowledge of the dynamics, for instance by using asymptotic relationships between local and implied volatilities for small time-to-maturities and near-ATM strikes, we could guess the optimal parameters once we know the mismatch between target and model implied quantities. A similar approach is also discussed in~\cite{Reghai2012} where the level of the local-volatility function is iteratively updated by the calibration procedure.

In the next section we calculate the level and the skew of the local volatility function in term of the model implied volatilities, and we use them to implement our iterative calibration strategy.

\subsection{Iterative Calibration Strategy}

We can relate the price of plain vanilla options obtained with dynamics~\eqref{eq:normspot} and the price obtained in a Black framework. We introduce the Black formula $c^{\rm BS}(t,k,\sigma)$, depending on time-to-maturities $t$, strikes $k$ and volatilities $\sigma$, as given by
\begin{equation}
c^{\rm BS}(t,k,\sigma) := \Phi(y+\sigma\sqrt{t})-k\Phi(y)
\;,\quad
y := -\frac{1}{\sigma\sqrt{t}}\log k -\frac{1}{2}\sigma \sqrt{t}
\label{eq: black}
\end{equation}%
where $\Phi$ is the cumulative Gaussian distribution, and all prices are normalized w.r.t.\ the forward price. Then, we can write
\Eq{
c^{\rm BS}(t,k,\sigma(t,k)) := c(t,k)
}%
as a definition for the implied Black volatility $\sigma(t,k)$ for $t>0$.

On the other hand, the extended Dupire Equation~\eqref{eq:dupire} obtained in Proposition~\ref{prop:dupire} can be solved for the local volatility function $\eta$ as given by
\begin{equation}
\eta^2(t,k) = 2 \,\frac{ (\partial_t + a(t) + a(t)(1-k)\,\partial_k) c^{\rm BS}(t,k,\sigma(t,k))}{ k^2 \,\partial^2_k c^{\rm BS}(t,k,\sigma(t,k)) }
\label{eq:locvol}
\end{equation}%
where the call price is now calculated as a Black price.

Thanks to the above relationship we can compute the level and skew of the local volatility function in term of the model implied volatilities. In particular, we are interested in the limit of small maturities and strike prices around the at-the-money value. The explicit dependence of the local volatility on the model implied volatility is obtained by computing the right-hand side of~\eqref{eq:locvol} by means of the chain rule and then taking the leading terms. In order to do so we have to extend the results of~\cite{Berestycki2002} to mean-reverting local volatility processes, which ensure the existence of the limit of the implied volatility function as the time to maturity approaches zero. We have the following proposition.

\begin{proposition}{\bf (Level and Skew of the Local Volatility Function)}
If the normalized spot price follows the dynamics of Proposition~\ref{prop:dupire}, the following limit exists:
\begin{equation}
\sigma(0,k) := \lim_{t \rightarrow 0} \sigma(t,k) = \log k \left( \int_1^k \frac{dx}{x \eta(0,x)} \right)^{-1}
\label{eq:harmonic}
\end{equation}
independently of the value of the mean-reversion speed. Moreover, the level and the skew of the local volatility function in the limit of zero time-to-maturity and at-the-money strike are related to the level and the skew of the implied Black volatility as given by
\begin{equation}
\eta(0,1) = \sigma(0,1)
\;,\quad
\partial_k \eta(0,1) = 2 \partial_k \sigma(0,1)
\label{eq:asymptotics}
\end{equation}
\label{prop:asymptotics}
\end{proposition}
\begin{proof}
The proof is shown in Appendix~\ref{sec:locimplvol}.
\end{proof}

The results given by Proposition~\ref{prop:asymptotics} can be used to tune the iteration algorithm of the calibration procedure to produce a quick convergence. We describe below our calibration strategy.
\begin{enumerate}
\item We consider the mean-reversion speed $a$ to be given. We illustrate in the following section some strategies we may adopt to fix its value.
\item Then we consider the local volatility function, and we set the nodes of the local-volatility function $\eta^{(0)}$ to be equal to the market volatilities.
\item We calculate the spot volatilities $\sigma^{(0)}$ implied from the PDE~\eqref{eq:dupire} when the local-volatility is interpolated on nodes $\eta^{(0)}$.
\item We compare $\sigma^{(0)}$ with $\sigma^{\rm mkt}$, and, if the distance is within a given threshold we stop the algorithm. Otherwise the algorithm proceeds to the next step.
\item We adjust the local volatility parameters by taking into account the first-order corrections for small time-to-maturities and near-ATM strikes to volatility level and skew derived in Proposition~\ref{prop:asymptotics}, as given by
\Eq{
\eta^{(1)}_{ij} = \eta^{(0)}_{ij} \frac{\sigma^{\rm mkt}_{ij_{\rm atm}}}{\sigma^{(0)}_{ij_{\rm atm}}} + 2 \left( \frac{\partial \sigma^{\rm mkt}_{ij}}{\partial k} - \frac{\partial \sigma^{(0)}_{ij}}{\partial k} \right) \Delta k_j \ind{j\neq j_{\rm atm}}
}%
where $j_{\rm atm}$ is the strike index referring to at-the-money options.
\item We repeat the algorithm from the third step with $\eta^{(1)}$. Then, the iteration is repeated until convergence, or until a maximum number of iterations is reached.
\end{enumerate}

Our calibration strategy can be viewed as a fixed-point algorithm. We can greatly increase the convergence rate by means of the Anderson Accelerated (AA) fixed-point algorithm, as described in~\cite{Anderson1965}, and later reviewed in~\cite{Walker2011}. The AA algorithm improves the standard fixed-point algorithm by taking into account on each step the errors measured in the previous steps. We describe the algorithm by means of the following pseudo-code.
\medskip
\begin{algorithmic}[1]
 \Procedure{AA}{$x^0$,$f$,$\epsilon$,$m$}
 \State $i \leftarrow 1$
 \State $x^1 \leftarrow f(x^0)$
 \Repeat
  \State $m_i \leftarrow \min\{m,i\}$
  \State $f^{(i)} \leftarrow \{f(x^{i-m_i}),\ldots,f(x^i)\}$
  \State $e^{(i)} \leftarrow \{f(x^{i-m_i})-x^{i-m_i},\ldots,f(x^i)-x^i\}$
  \State $\alpha^{(i)} \leftarrow \arg\min\{ \| \alpha^{(i)} \cdot e^{(i)} \| \}$ subject to $\sum_{j=0}^{m_i} \alpha^{(i)}_j = 1$
  \State $x^{i+1} \leftarrow \alpha^{(i)} \cdot f^{(i)}$
  \State $i \leftarrow i+1$
 \Until{$\| x^{i+1} - x^i \| < \epsilon$}
 \State {\bf return} $x^{i+1}$
 \EndProcedure
\end{algorithmic}
\medskip

We present some numerical experiments with real market data. We find that the AA scheme, along with the choice of correcting both the level and the skew of the local-volatility function on each iteration step, greatly improves the algorithm of~\cite{Reghai2012}, leading to a calibration error in volatility space of one tenth of basis point within 15-30 iterations in most cases. We show in Figures~\ref{fig:calibcoffee},~\ref{fig:calibcopper}, and~\ref{fig:calibwti} the convergence rate. The blue lines refer to a calibration strategy where both level and skew of the local-volatility function are updated, while red lines to the one where only the level is updated. Dashed lines refer standard fixed-point iteration, while solid lines to AA iterations. Our target calibration strategy is the solid blue line, while the calibration strategy of~\cite{Reghai2012} is the dashed red line.

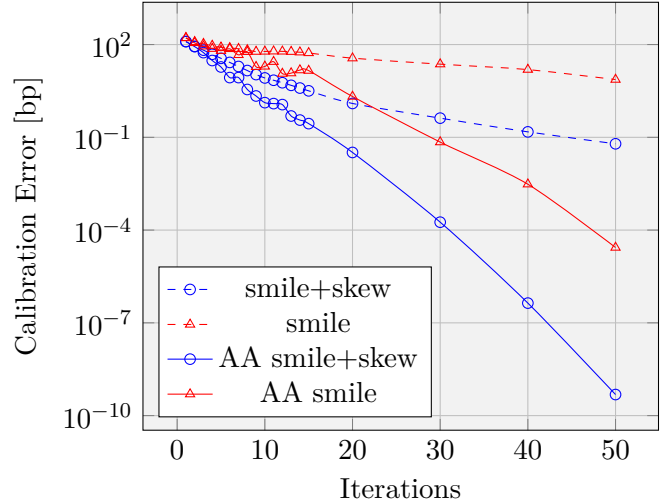
\begin{figure}
\begin{center}
\begin{tikzpicture}
\begin{semilogyaxis}[xlabel=Iterations,
                     ylabel=Calibration Error {[bp]},
                     ylabel style={overlay},
                     grid=major,
                     legend style={legend pos=south west},
                     axis background/.style={fill=gray!10}]
  \addplot [color=blue,mark=o,mark options={style=solid},smooth,style=dashed] table [y=fe,x=N] from \calibcoffee;
  \addplot [color=red,mark=triangle,mark options={style=solid},smooth,style=dashed] table [y=rbv,x=N] from \calibcoffee;
  \addplot [color=blue,mark=o,smooth] table [y=AAfe,x=N] from \calibcoffee;
  \addplot [color=red,mark=triangle,smooth] table [y=AArbv,x=N] from \calibcoffee;
  \legend{smile+skew,smile,AA smile+skew,AA smile}
\end{semilogyaxis}
\end{tikzpicture}
\end{center}
\caption{Calibration of 65 options on Coffee Arabica Futures quoted on 2 November 2017 on ICE market with and without Anderson acceleration scheme (AA). Dashed red curve refers to~\cite{Reghai2012} algorithm, solid blue line to the algorithm presented in this paper with the speed of mean reversion set to zero.}
\label{fig:calibcoffee}
\end{figure}

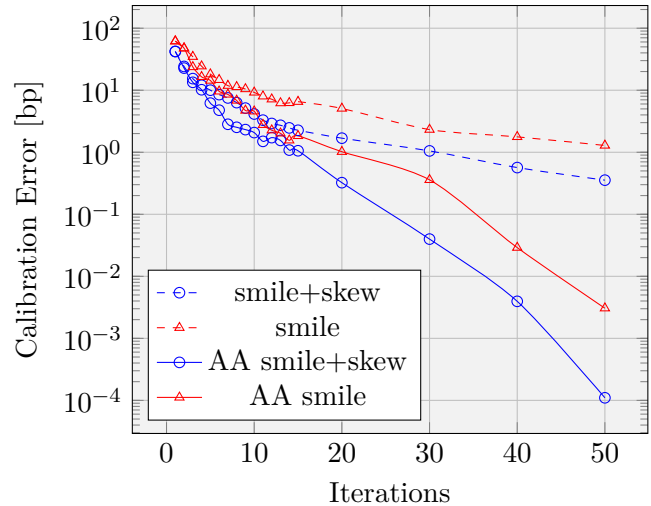
\begin{figure}
\begin{center}
\begin{tikzpicture}
\begin{semilogyaxis}[xlabel=Iterations,
                     ylabel=Calibration Error {[bp]},
                     ylabel style={overlay},
                     grid=major,
                     legend style={legend pos=south west},
                     axis background/.style={fill=gray!10}]
  \addplot [color=blue,mark=o,mark options={style=solid},smooth,style=dashed] table [y=fe,x=N] from \calibcopper;
  \addplot [color=red,mark=triangle,mark options={style=solid},smooth,style=dashed] table [y=rbv,x=N] from \calibcopper;
  \addplot [color=blue,mark=o,smooth] table [y=AAfe,x=N] from \calibcopper;
  \addplot [color=red,mark=triangle,smooth] table [y=AArbv,x=N] from \calibcopper;
  \legend{smile+skew,smile,AA smile+skew,AA smile}
\end{semilogyaxis}
\end{tikzpicture}
\end{center}
\caption{Calibration of 125 options on Copper Futures quoted on 2 November 2017 on CME market with and without Anderson acceleration scheme (AA). Dashed red curve refers to~\cite{Reghai2012} algorithm, solid blue line to the algorithm presented in this paper with the speed of mean reversion set to zero.}
\label{fig:calibcopper}
\end{figure}

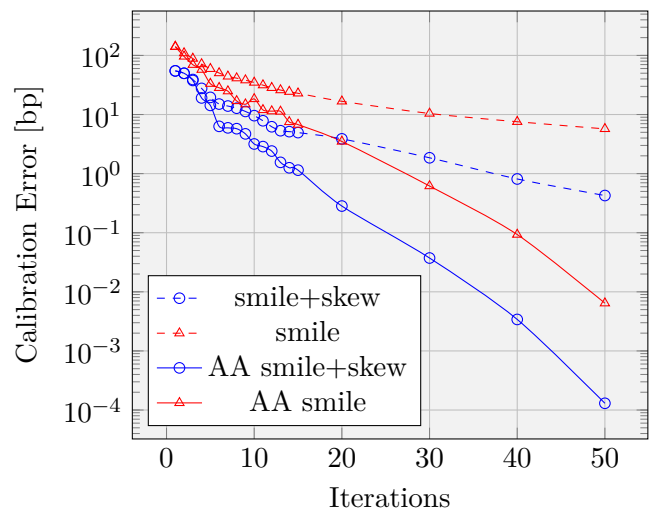
\begin{figure}
\begin{center}
\begin{tikzpicture}
\begin{semilogyaxis}[xlabel=Iterations,
                     ylabel=Calibration Error {[bp]},
                     ylabel style={overlay},
                     grid=major,
                     legend style={legend pos=south west},
                     axis background/.style={fill=gray!10}]
  \addplot [color=blue,mark=o,mark options={style=solid},smooth,style=dashed] table [y=fe,x=N] from \calibwti;
  \addplot [color=red,mark=triangle,mark options={style=solid},smooth,style=dashed] table [y=rbv,x=N] from \calibwti;
  \addplot [color=blue,mark=o,smooth] table [y=AAfe,x=N] from \calibwti;
  \addplot [color=red,mark=triangle,smooth] table [y=AArbv,x=N] from \calibwti;
  \legend{smile+skew,smile,AA smile+skew,AA smile}
\end{semilogyaxis}
\end{tikzpicture}
\end{center}
\caption{Calibration of 100 options on WTI Crude Oil Futures quoted on 23 May 2018 on CME market with and without Anderson acceleration scheme (AA). Dashed red curve refers to~\cite{Reghai2012} algorithm, solid blue line to the algorithm presented in this paper with the speed of mean reversion set to zero.}
\label{fig:calibwti}
\end{figure}

In the previous examples the speed of mean reversion is set equal to zero. In Figure~\ref{fig:calibwti-a} we show the convergence results for WTI calibration for different values of the speed of mean reversion, and the resulting at-the-money local volatility functions along with the market volatilities. As previously stated in Remark~\ref{remark:meanrev} we recall that with specific market conditions, e.g.\ when volatility decreases strongly over time, the calibration procedure may fail with small values of $a$, and in particular with $a=0$. In such case, the introduction of a mean reversion is necessary to match the market volatility.

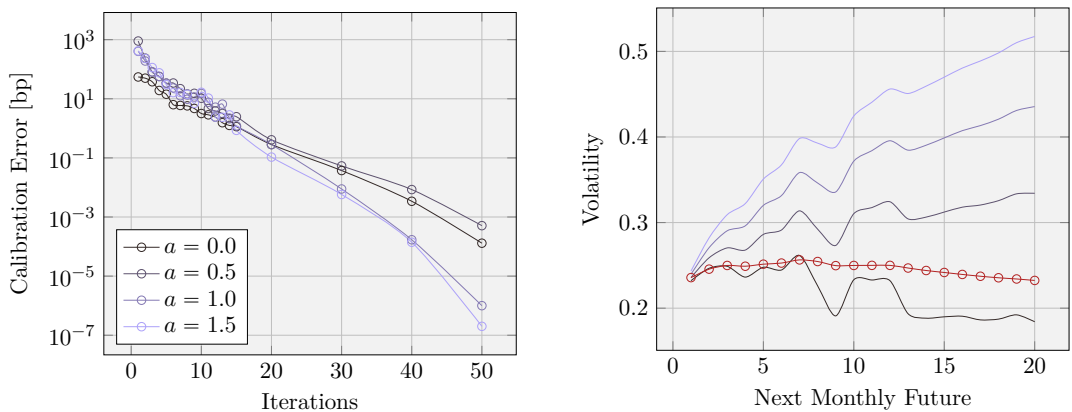
\begin{figure}
\begin{center}
\begin{tikzpicture}[scale=0.8]
\begin{semilogyaxis}[xlabel=Iterations,
                     ylabel=Calibration Error {[bp]},
                     ylabel style={overlay},
                     grid=major,
                     legend style={legend pos=south west},
                     axis background/.style={fill=gray!10}]
  \addplot [color=myblueA,mark=o,smooth] table [y=a0,x=N] from \calibwtia;
  \addplot [color=myblueB,mark=o,smooth] table [y=a5,x=N] from \calibwtia;
  \addplot [color=myblueC,mark=o,smooth] table [y=a10,x=N] from \calibwtia;
  \addplot [color=myblueD,mark=o,smooth] table [y=a15,x=N] from \calibwtia;
  \legend{$a=0.0$,$a=0.5$,$a=1.0$,$a=1.5$}
\end{semilogyaxis}
\end{tikzpicture}
\hspace*{1cm}
\begin{tikzpicture}[scale=0.8]
\begin{axis}[xlabel=Next Monthly Future,
                     ylabel=Volatility,
                     ylabel style={overlay},
                     grid=major,
                     axis background/.style={fill=gray!10}]
  \addplot [color=myblueA,smooth] table [y=a0,x=T] from \calibwtilv;
  \addplot [color=myblueB,smooth] table [y=a5,x=T] from \calibwtilv;
  \addplot [color=myblueC,smooth] table [y=a10,x=T] from \calibwtilv;
  \addplot [color=myblueD,smooth] table [y=a15,x=T] from \calibwtilv;
  \addplot [color=mydarkred,mark=o,smooth] table [y=mkt,x=T] from \calibwtilv;
\end{axis}
\end{tikzpicture}
\end{center}
\caption{Calibration of 100 options on WTI Crude Oil Futures quoted on 23 May 2018 on CME market with the speed of mean reversion ranging from zero up to $1.5$ with a step of $0.5$. On the left side panel we show the calibration error, while on the right side panel the resulting at-the-money local volatilities (in blue) along with market volatilities (in red).}
\label{fig:calibwti-a}
\end{figure}

\subsection{Calibration of the Mean-Reversion Speed}
\label{sec:meanrev}

If we look only at plain-vanilla prices we cannot fix the mean-reversion speed parameter in a robust way. Indeed, we can easily calibrate the local volatility function to plain-vanilla option prices for different choices of mean-reversion speed, as we can be easily seen in the left side panel of Figure~\ref{fig:calibwti-a}. We need a different strategy to fix this parameter.

A natural solution is to enrich our calibration set. For instance, we can look at the prices of mid-curve options (MCO) and calendar spread options (CSO), which are usually quoted for commodity markets. Here, the main problems concern the fact that these quotes are not always liquid. Moreover, the general specification of the model inclusive of curve and smile dynamics, which we are discussing in the next session, leads to time-consuming pricing algorithms not suited for calibration purposes. Thus, a possible strategy to fix the mean reversion speed is following~\cite{Drimus2013} where a similar problem occurring in the equity market is tackled by requiring that the mean-reversion speed is fixed so that the slope of the local-volatility function is minimized. We leave for a future work the complete exploration of this issue.

However, if MCO or CSO prices are liquid enough, we can derive a simple calibration algorithm if we limit ourselves to model the futures prices by means of the the fictitious spot price. Indeed, in this case we can use the spot-futures relationship~\eqref{eq:fwdspot} to map plain-vanilla options on irregular expiry dates and spread options on futures prices as standard plain-vanilla options on the fictitious spot price. This approach is similar in spirit to the one taken by~\cite{andersen2010interest} when deriving a local-volatility model for swaption prices.

MCO's on commodity futures are European call options where the expiry date of the option contract occurs well before the $T^{\textrm last}$ of the underlying futures contract. Thus, the pricing formula is simply given by Equations~\eqref{eq:optmarg} and~\eqref{eq:optcoll} with an option expiring date which precedes the expiry date of a standard plain-vanilla option.

CSO's on commodity futures are spread options between two different futures contracts on the same commodity, namely with payoff
\Eq{
\left(F_{T_e}^{\textrm{ID}_1} - F_{T_e}^{\textrm{ID}_2}- K\right)^{\!+}
}%
where $T_e$ is the spread option exercise date, while $\textrm{ID}_1$ and $\textrm{ID}_2$ are the identifiers of the two futures. If the fictitious spot dynamics holds, by using Equation~\eqref{eq:fwdspot}, we can write for the future-style margining case\footnote{Here, we use the date subscripts to distinguish the futures $T^{\textrm{last}}$ dates, not as in Section~\ref{sec:market} to distinguish between the start and end dates of a period.}
\Eq{
\begin{aligned}
V^{\rm CSO}_t(T_e,T_1^{\textrm{last}},T_2^{\textrm{last}},K) :&= \Ex{t}{\left(F_{T_e}(T_{1}^{\textrm{last}}) - F_{T_e}(T_{2}^{\textrm{last}}) - K\right)^{\!+}}  \\ 
&=
\begin{cases}
A \,c(T^e,B)                                &  \text{if} \;\; A>0, B>0 \\
A \,(1-B)                                      &  \text{if} \;\; A>0, B\leq0 \\
-A \left( \,c(T^e,B)+B-1 \right)  &  \text{if} \;\; A<0, B>0 \\
0                                                  &  \text{if} \;\; A\leq0, B\leq0 \\
\end{cases}
\end{aligned}
}%
where we define
\Eq{
A(T_e,T_1^{\textrm{last}},T_2^{\textrm{last}}) := F_0(T_{1}^{\textrm{last}}) e^{-\int_{T_e}^{T_{1}^{\textrm{last}}} a(u) \,du} - F_0(T_{2}^{\textrm{last}}) e^{-\int_{T_e}^{T_{2}^{\textrm{last}}} a(u) \,du}
}%
and
\Eq{
B(T_e,T_1^{\textrm{last}},T_2^{\textrm{last}}) := 1 + \frac{K - F_0(T_{1}^{\textrm{last}}) + F_0(T_{2}^{\textrm{last}})}{A(T_e,T_1^{\textrm{last}},T_2^{\textrm{last}})}
}%

We tested these results with real market data. We assume a time-independent speed of mean reversion $a$ and we consider the case of the WTI market. In particular, we consider CSO's on two consecutive futures, which we can term CL1 and CL2. The maturities of such futures occur on two consecutive months. We can read from the plain-vanilla market the implied volatilities $\sigma^1_1$ and $\sigma^2_2$ of such futures. On the other hand, the market quotes CSO prices, which we can express in a normalized way. We assume, only for quotation purpose, that the futures prices follow a log-normal dynamics and we set the correlations among them to one, so that CSO prices result in a mono-dimensional integral.
\begin{multline}
V^{\rm CSO}_t(T_e,T_1^{\textrm{last}},T_2^{\textrm{last}},K) =\\ \int_{-\infty}^{\infty} \left( F_0(T_1^{\textrm{last}}) e^{-(\sigma_1^1)^2 T_e - \sigma_1^1 \sqrt{T_e}x} - F_0(T_2^{\textrm{last}}) e^{-(\sigma_1^2)^2 T_e - \sigma_1^2 \sqrt{T_e}x} - K \right)^{\!+} \phi(x) \,dx
\end{multline}%
where $T_e$ is the CSO expiry date, $T_1^f$ is the last trading date of the CL1, $T_2^f$ is the last trading date of the CL2, $K$ is the CSO strike price and $\phi(x)$ is the density function of a standard normal variable. Once the market prices of CSO are known we can invert\footnote{The equation actually shows two positive solutions: We arbitrarily choose the lesser one to define our quotation metric.} the above equation to derive $\sigma^2_1$. In Figure~\ref{fig:csowti} we plot the volatility drop implied by calendar spread options for two consecutive futures contracts, namely the difference $\sigma^2_2-\sigma^2_1$. We can see that a value of $a=0.5$ fit most of the market quotes. We could promote $a$ to be time-dependent if we wish a better or exact fit.

\begin{figure}
\begin{center}
\begin{tikzpicture}
\begin{axis}[xlabel=Expiry,
                     ylabel=Volatility Drop {[bp]},
                     ylabel style={overlay},
                     grid=major,
                     axis background/.style={fill=gray!10}]
  \addplot [color=myblueA,smooth] table [y=a0,x=T] from \csowti;
  \addplot [color=myblueB,smooth] table [y=a5,x=T] from \csowti;
  \addplot [color=myblueC,smooth] table [y=a10,x=T] from \csowti;
  \addplot [color=myblueD,smooth] table [y=a15,x=T] from \csowti;
  \addplot [color=mydarkred,mark=o,mark options={style=solid},style=solid] table [y=market,x=T] from \csowti;
\end{axis}
\end{tikzpicture}
\end{center}
\caption{WTI Crude Oil 1 Month Calendar Spread Options quoted on 23 May 2018 on CME market. The red dots represent the first 18 quoted maturities in terms of volatility drop, while the blue lines are model-implied volatility drops for mean reversion values ranging from top to bottom from $1.5$ to zero in step of $0.5$.}
\label{fig:csowti}
\end{figure}
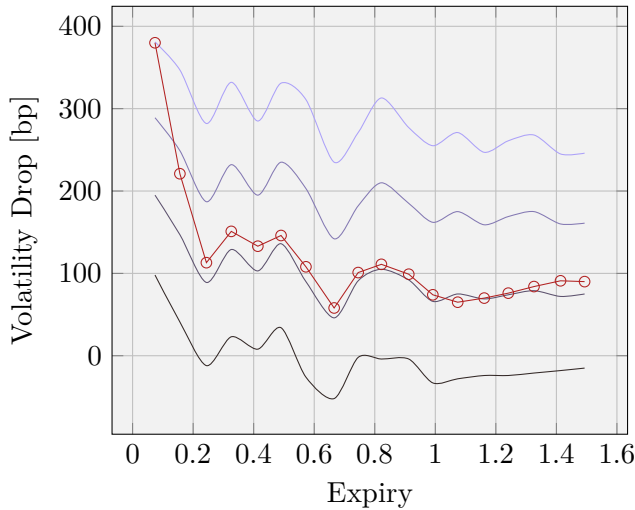

\section{Curve and Smile Dynamics}
\label{sec:smile}

By using the normalized spot process we are able to price all futures options with a single PDE evaluation. In this way we implicitly obtain the marginal probability densities of futures prices. We now look at the joint probability densities between two or more futures prices, and at the transition densities in each futures dynamics. In this way we will be able to model the curve and skew dynamics.

\subsection{Markovian Projections}

A more flexible description of joint probability densities and transition densities can be achieved by introducing new risk factors. Usually a complex curve dynamics can be modelled by allowing each futures price to be driven by its own stochastic process, while smile dynamics can be modelled by means of a stochastic volatility process. In particular, the term SLV model is used to refer to models with a volatility depending both on the spot price and on additional stochastic processes. Here, we extend the stochastic-local volatility framework to allow also for curve dynamics. SLV models were first presented in~\cite{Ren2007} along with the description of a practical calibration procedure based on the application of the Gy\"ongy Lemma, see~\cite{Gyongy1986}, which states under which conditions the marginal densities of a stochastic-local volatility model match the ones predicted by a local-volatility model.

Our intent is to promote each futures price to be driven by a different SLV model, and, at the same time, to limit the number of free parameters since commodity markets quotes few derivative contracts beside the plain-vanilla options already used when we have calibrated the local-volatility model. We start by considering a generic stochastic volatility dynamics for futures prices
\begin{equation}
dF_t(T) = \nu_t(T) \cdot dW^F_t
\label{eq:future}
\end{equation}%
where $F$ is the vector of futures prices (each entry refers to a different quoted maturity), $\nu$ is an adapted vector process, and $W^F$ a vector of standard Brownian motions under risk-neutral measure. The notation $a \cdot b$ refers to the internal product between vectors $a$ and $b$.

We are interested in preserving the marginal densities of futures prices, and not of the normalized spot price, so that we apply the Gy\"ongy Lemma directly on the futures dynamics. Thus, we can calibrate the SLV model to match the marginal densities of futures prices predicted by the local-volatility model, and in turn to match the plain-vanilla option-on-futures prices quoted on the market, by requiring that $\nu$ satisfies
\begin{equation}
\ExC{0}{\|\nu_t(T)\|^2}{F_t(T)=K} = \eta_F^2(t,T,K)
\label{eq:gyongy}
\end{equation}%
where the local-volatility for futures prices in defined in Equation~\eqref{eq:etaF}.

We recall that following this approach we cannot use any longer the results of Section~\ref{sec:meanrev} since we are altering the futures dynamics keeping fixed only the marginal densities, so that the dependency structure of future rates may change. Thus, if we proceed in this direction we have also to design a new calibration procedure for the mean reversion parameter.

\subsection{Multiple Driving Factors}

As we have seen in Section~\ref{sec:meanrev} the one-factor mean-reverting normalized spot dynamics of Equation~\eqref{eq:normspot} can be used to recover CSO market quotes, and in turn to describe how the futures volatility drops near the futures expiry date. On the other hand, if we are interested in modelling the dependency between futures, for instance we are interested in terminal correlations to price options on the spread between a short-dated and a long-dated futures contract, we should introduce more driving factors.

Indeed, if we calculate the terminal covariance between two futures prices in the one-factor mean-reverting local volatility model we get by using Equation~\eqref{eq:fwdspot}
\Eq{
\cov{0}{}{F_t(T_1)}{F_t(T_2)} = F_0(T_1)F_0(T_2) \left( \left(\Ex{0}{s^2_t}-1\right) e^{-\int_t^{T_1} a(u) \,du} e^{-\int_t^{T_2} a(u) \,du} - 1 \right)
}%
leading to
\Eq{
\corr{0}{}{F_t(T_1)}{F_t(T_2)} := \frac{\cov{0}{}{F_t(T_1)}{F_t(T_2)}}{\sqrt{\var{0}{}{F_t(T_1)}\var{0}{}{F_t(T_2)}}} = 1
}%
This means that in the one-factor model CSO market quotes are recovered thanks to the volatility profile of each futures contract, and not to a specific correlation structure given by the model. Thus, if we need to model terminal correlations we need more driving factors.

However, if we wish to preserve the calibration to plain vanilla options, we need to define a dynamics for futures prices which satisfies Equation~\eqref{eq:gyongy}. A simple way to do so is considering a two-dimensional model by defining the volatility vector process as
\begin{equation}
\nu_t(T) := \eta_F(t,T,F_t(T)) \begin{bmatrix} \rho(T) \\ \sqrt{1-\rho(T)^2} \end{bmatrix}
\end{equation}%
where $\rho(T)$ is a function of futures expiry date and it is bounded in the interval $[-1,1]$. With this choice the instantaneous correlation among two different futures is given by
\Eq{
d\langle F(T_1),F(T_2) \rangle_t = \left( \rho(T_1)\rho(T_2) + \sqrt{(1-\rho^2(T_1))(1-\rho^2(T_2))} \right) dt
}%
and the terminal correlations are no longer naive.

The calibration of mean-reversion and correlation parameters to market quotes can be achieved by solving the corresponding bi-dimensional pricing PDEs.

\subsection{Stochastic Volatility Extensions}

If we wish to model also the smile dynamics we need to introduce a stochastic volatility process. A possible simple prescription is given by the two-dimensional process
\begin{equation}
\nu_t(T) := v_t \frac{\eta_F(t,T,F_t(T))}{\sqrt{\ExC{0}{v^2_t}{F_t(T)}}} \begin{bmatrix} \rho(T) \\ \sqrt{1-\rho(T)^2} \end{bmatrix}
\label{eq:nu}
\end{equation}%
where the volatility process $v_t$ can be defined under risk-neutral measure as in~\cite{Ren2007} as given by
\begin{equation}
d\log v_t = - \left(\frac{1+e^{-2t}}{2}\xi^2 + \log v_t\right)\,dt + \xi \,dW^v_t
\;,\quad
v_0 = 1
\;,\quad
d\langle W^F, W^v\rangle_t = \begin{bmatrix} \rho^v \\ \rho^v \end{bmatrix} dt
\label{eq:v}
\end{equation}%
where the drift part is defined so that $\Ex{0}{v^2_t}=1$, while the correlation parameter $\rho^v$ is bounded in $\left[-\frac{1}{\sqrt{2}},\frac{1}{\sqrt{2}}\right]$. Other prescriptions for the volatility process can be analyzed such as the classical CIR model, or more recent proposal such as~\cite{Tataru2010} or the polynomial models of~\cite{Ackerer2016b}.

The vol-of-vol $\xi$ and the correlations $\{\rho(T_1),\ldots,\rho(T_n)\}$ and $\rho^v$ are free parameters which can be calibrated to exotic options, such as options on futures of different delivery tenors, spread options, MCO's or CSO's. Moreover, when considering the full SLV specification, the speed of mean reversion can be calibrated along with these parameters. Notice, that here CSO's cannot easily be mapped into plain-vanilla options to spare computational time. We leave to a future work the design of calibration procedures dealing with the three-factor dynamics deriving from Equation~\eqref{eq:nu}.

Monte Carlo simulation of the model dynamics, when coefficients depend on conditional densities, can be designed as in~\cite{Guyon2012} or in~\cite{Oosterlee2014}. Both methods starts from the knowledge of the local-volatility function, so that they do not require an explicit calibration step of plain vanilla options.

\section{Conclusion and Further Developments}
\label{sec:conclsion}

In this paper we presented a SLV model for derivative contracts on commodity futures inclusive of forward-curve and smile dynamics characterized by a parsimonious parametrization to deal with the limited number of options quoted in the market. In particular, we first introduced a local-volatility one-factor process with affine drift to model future-price marginal densities, we extended the Dupire equation to allow the implementation of a fast and robust calibration algorithm based on a fixed-point iteration accelerated by means of the Anderson scheme. We investigate the performance of this one-factor model to recover exotic options quoted in the market, such as calendar spread options. Then, we defined a SLV dynamics for each futures price, so to model its curve and smile dynamics, by matching the future-price marginal densities calibrated in the first step.

We leave for future developments an extensive numerical analysis of predicted options prices in different commodity markets, the further extension of our version of the Dupire equation to the case of early-exercise options, and the implementation of calibration algorithms to fix correlation and volatility-of-volatility parameters.

\newpage

\bibliographystyle{plainnat}
\bibliography{cosmile}

\begin{thebibliography}{27}
\providecommand{\natexlab}[1]{#1}
\providecommand{\url}[1]{\texttt{#1}}
\expandafter\ifx\csname urlstyle\endcsname\relax
  \providecommand{\doi}[1]{doi: #1}\else
  \providecommand{\doi}{doi: \begingroup \urlstyle{rm}\Url}\fi

\bibitem[Ackerer and D.(2016)]{Ackerer2016a}
D.~Ackerer and Filipovic D.
\newblock Linear credit risk models.
\newblock \emph{SFI research paper series}, 34\penalty0 (16), 2016.

\bibitem[Ackerer et~al.(2016)Ackerer, Larsson, and D.]{Ackerer2016b}
D.~Ackerer, M.~Larsson, and Filipovic D.
\newblock The jacobi stochastic volatility model.
\newblock \emph{SFI research paper series}, 34\penalty0 (16), 2016.

\bibitem[Albani et~al.(2017)Albani, Ascher, and Zubelli]{Albani2017}
V.~Albani, U.~Ascher, and J.~Zubelli.
\newblock Local volatility models in commodity markets and online calibration.
\newblock \emph{The Journal of Computational Finance}, \penalty0 (21), 2017.

\bibitem[Andersen and Piterbarg(2010)]{andersen2010interest}
L.B.G. Andersen and V.V. Piterbarg.
\newblock \emph{{Interest Rate Modeling}}.
\newblock Atlantic Financial Press, 2010.

\bibitem[Anderson(1965)]{Anderson1965}
D.~Anderson.
\newblock Iterative procedures for nonlinear integral equations.
\newblock \emph{Journal of the {ACM}}, 4\penalty0 (12):\penalty0 547--560,
  1965.

\bibitem[Bentata and Cont(2015)]{Bentata2015}
A.~Bentata and R.~Cont.
\newblock Forward equations for option prices in semimartingale models.
\newblock \emph{Finance and Stochastics}, 19\penalty0 (3):\penalty0 617--651,
  2015.

\bibitem[Berestycki et~al.(2002)Berestycki, Busca, and Florent]{Berestycki2002}
H.~Berestycki, J.~Busca, and I.~Florent.
\newblock Asymptotics and calibration of local volatility models.
\newblock \emph{Quantitative Finance}, 2\penalty0 (1):\penalty0 61--69, 2002.

\bibitem[Brigo et~al.(2013)Brigo, Morini, and
  Pallavicini]{BrigoMoriniPallavicini2012}
D.~Brigo, M.~Morini, and A.~Pallavicini.
\newblock \emph{Counterparty Credit Risk, Collateral and Funding with pricing
  cases for all asset classes}.
\newblock Wiley, Chichester, 2013.

\bibitem[Chiminello(2015)]{Chiminello2015}
F.~Chiminello.
\newblock Oil goes local.
\newblock Talk at Imperial College London, 2015.

\bibitem[De~Marco and Henry-Labord\'ere(2017)]{DeMarco2016}
S.~De~Marco and P.~Henry-Labord\'ere.
\newblock Local volatility from american options.
\newblock \emph{Risk Magazine}, \penalty0 (11), 2017.

\bibitem[Deelstra and Ray\'ee(2013)]{Deelstra2013}
G.~Deelstra and G.~Ray\'ee.
\newblock Local volatility pricing models for long-dated fx derivatives.
\newblock \emph{Applied Mathematical Finance}, 4\penalty0 (20):\penalty0
  380--402, 2013.

\bibitem[Derman and Kani(1994)]{Derman1994}
E.~Derman and I.~Kani.
\newblock Riding on a smile.
\newblock \emph{Risk Magazine}, \penalty0 (7):\penalty0 32--39, 1994.

\bibitem[Drimus and Farkas(2013)]{Drimus2013}
G.~Drimus and W.~Farkas.
\newblock Local volatility of volatility for the {VIX} market.
\newblock \emph{Review of Derivatives Research}, 3\penalty0 (16):\penalty0
  267--293, 2013.

\bibitem[Dupire(1994)]{Dupire1994}
B.~Dupire.
\newblock Pricing with a smile.
\newblock \emph{Risk Magazine}, \penalty0 (1):\penalty0 18--20, 1994.

\bibitem[Gatheral(2006)]{gatheral2006volatility}
J.~Gatheral.
\newblock \emph{{The Volatility Surface: A Practitioner's Guide}}.
\newblock Wiley, 2006.

\bibitem[Guyon and Henry-Labord\'ere(2012)]{Guyon2012}
J.~Guyon and P.~Henry-Labord\'ere.
\newblock Pricing with a smile.
\newblock \emph{Risk Magazine}, \penalty0 (1):\penalty0 88--93, 2012.

\bibitem[Gy{\"o}ngy(1986)]{Gyongy1986}
I.~Gy{\"o}ngy.
\newblock Mimicking the one-dimensional marginal distributions of processes
  having an it{\^o} differential.
\newblock \emph{Probability Theory and Related Fields}, \penalty0
  (71):\penalty0 501--516, 1986.

\bibitem[Lee(2004)]{Lee2004}
R.~Lee.
\newblock The moment formula for implied volatility at extreme strikes.
\newblock \emph{Mathematical Finance}, 14\penalty0 (3):\penalty0 469--480,
  2004.

\bibitem[Moreni and Pallavicini(2017)]{Moreni2017}
N.~Moreni and A.~Pallavicini.
\newblock {Derivative pricing with collateralization and FX market
  dislocations}.
\newblock \emph{International Journal of Theoretical and Applied Finance},
  6\penalty0 (20), 2017.

\bibitem[Pilz and Schl{\"o}gl(2011)]{Pilz2011}
K.~Pilz and E.~Schl{\"o}gl.
\newblock A hybrid commodity and interest rate market model.
\newblock \emph{Quantitative Finance}, 4\penalty0 (13):\penalty0 543--560,
  2011.

\bibitem[Protter(2005)]{protter}
P.~Protter.
\newblock \emph{{Stochastic integration and differential equations}}.
\newblock Springer-Verlag, 2005.

\bibitem[Reghai et~al.(2012)Reghai, Boya, and Vong]{Reghai2012}
A.~Reghai, G.~Boya, and G.~Vong.
\newblock Local volatility: smooth calibration and fast usage.
\newblock \emph{Working Paper}, 2012.
\newblock \doi{10.2139/ssrn.2008215}.
\newblock URL \url{https://ssrn.com/abstract=2008215}.

\bibitem[Ren et~al.(2007)Ren, Madan, and Qian~Qian]{Ren2007}
Y.~Ren, D.~Madan, and M.~Qian~Qian.
\newblock Calibrating and pricing with embedded local volatility models.
\newblock \emph{Risk Magazine}, \penalty0 (9):\penalty0 138--143, 2007.

\bibitem[Swishchuk(2008)]{Swishchuk2008}
A.~Swishchuk.
\newblock Explicit option pricing formula for a mean-reverting asset in energy
  market.
\newblock \emph{Journal of Numerical and Applied Mathematics}, 1\penalty0
  (96):\penalty0 216--233, 2008.

\bibitem[Tataru and Fisher(2010)]{Tataru2010}
G.~Tataru and T.~Fisher.
\newblock Stochastic local volatility.
\newblock Technical report, Bloomberg, 2010.

\bibitem[Van~der Stoep et~al.(2014)Van~der Stoep, Grzelak, and
  Oosterlee]{Oosterlee2014}
A.~Van~der Stoep, L.~Grzelak, and C.~Oosterlee.
\newblock The {H}eston stochastic-local volatility model: Efficient {M}onte
  {C}arlo simulation.
\newblock \emph{International Journal of Theoretical and Applied Finance},
  17\penalty0 (7), 2014.

\bibitem[Walker and Ni(2011)]{Walker2011}
H.~Walker and P.~Ni.
\newblock Anderson acceleration for fixed-point iterations.
\newblock \emph{SIAM Journal on Numerical Analysis}, 4\penalty0 (49):\penalty0
  1715--1735, 2011.

\end{thebibliography}

\appendix

\section{Proof of Proposition~\ref{prop:asymptotics}}
\label{sec:locimplvol}

In this appendix we show the proof of Proposition~\ref{prop:asymptotics}. As a first tool we need two lemmas.
\begin{lemma}
In the Black framework given by Equation~\eqref{eq: black} if the small-time limit
\Eq{
\sigma(0,k) := \lim_{t\rightarrow 0} \sigma(t,k)
}%
exists, then also the following limit exists
\begin{equation}
\lim_{t\rightarrow 0} \frac{\Phi(y + \sigma \sqrt{t}) - \Phi(y)}{\mathcal{V}^{\rm BS}(t,k,\sigma)} = \sigma(0,k)
\label{eq:limit}
\end{equation}
where the plain-vanilla Black Vega is defined as
\Eq{
\mathcal{V}^{\rm BS}(t,k,\sigma) := \partial_\sigma c^{\rm BS}(t,k,\sigma)
}%
\label{lemma:limit}
\end{lemma}

\begin{proof}
From the definition of cumulative Gaussian probability distribution we can write for any real numbers $a$ and $b$ such that $b>a$ and $ab>0$
\Eq{
\Phi(b) - \Phi(a) = \int_a^b \phi(x) \,dx
}%
where $\phi$ is the Gaussian probability density. We can limit the value of the integral by considering the minimum and the maximum of the integrand.
\Eq{
(b-a) \phi(\underline{x}) \leq \Phi(b) - \Phi(a) \leq (b-a) \phi(\widebar{x})
}%
where we define
\Eq{
\underline{x} := \argmin_{x\in[a,b]} \phi(x)
\;,\quad
\widebar{x} := \argmax_{x\in[a,b]} \phi(x)
}%
In particular if $b:=y+\sigma\sqrt{t}$ and $a:=y$, where $y$ is defined as in Equation~\eqref{eq: black}, we can write
\Eq{
\sigma\sqrt{t} \,\phi(\underline{x}) \leq \Phi(y+\sigma\sqrt{t}) - \Phi(y) \leq \sigma\sqrt{t} \,\phi(\widebar{x})
}%
Moreover, since the plain-vanilla Black Vega can be explicitly calculated, and it is given by
\Eq{
\mathcal{V}^{\rm BS}(t,k,\sigma) = \sqrt{t} \,\phi(y+\sigma\sqrt{t})
}%
we can write
\Eq{
\sigma \frac{\phi(\underline{x})}{\phi(y+\sigma\sqrt{t})} \leq \frac{\Phi(y+\sigma\sqrt{t}) - \Phi(y)}{\mathcal{V}^{\rm BS}(t,k,\sigma)} \leq \sigma \frac{\phi(\widebar{x})}{\phi(y+\sigma\sqrt{t})}
}%
Then, since we have
\Eq{
\lim_{t\rightarrow 0} \frac{\phi(\underline{x})}{\phi(y+\sigma\sqrt{t})} = 1
\;,\quad
\lim_{t\rightarrow 0}\frac {\phi(\widebar{x})}{\phi(y+\sigma\sqrt{t})} = 1
}%
we prove the Lemma.
\end{proof}

\begin{lemma} {\bf (Comparison Principle)}
We consider two implied volatility functions $\widebar{\sigma}(t,k)$ and $\underline{\sigma}(t,k)$. If the corresponding local volatility functions obtained by means of Equation~\eqref{eq:locvol} are ordered so that $\eta[\underline{\sigma}](t,k) \le \eta[\widebar{\sigma}](t,k)$ for any time and strike, then $\underline{\sigma}(t,k) \le \widebar{\sigma}(t,k)$.
\label{lemma:comparison}
\end{lemma}

\begin{proof}
We write the Kolmogorov backward equation satisfied by a plain vanilla option price when the underlying risk factor follow the dynamics of Proposition~\ref{prop:dupire}. If we term $c$ the price at time $u$ with spot level $s$ for a call option with maturity $t$ and strike $k$ (we omit some dependencies to lighten the notation), namely we can write
\Eq{
\left( \partial_u + a(1-s) \partial_s + \frac{1}{2} \eta^2 s^2 \partial_s^2 \right) c = 0
}%
Then, if we use the local volatility functions $\eta[\underline{\sigma}]$ and $\eta[\widebar{\sigma}]$ defined in the hypotheses of the Lemma, we can write the PDE for the corresponding call option prices ${\underline c}$ and ${\widebar c}$, leading for their difference $w := {\widebar c} - {\underline c}$ to
\Eq{
\left( \partial_u + a(1-s) \partial_s + \frac{1}{2} \eta^2[\widebar{\sigma}] s^2 \partial_s^2 \right) w + \frac{1}{2} ( \eta^2[\widebar{\sigma}] - \eta^2[\underline{\sigma}] ) s^2 \partial_s^2 {\underline c} = 0
\;,\quad
w(t,k) = 0
}%
which can be interpreted via the Feynman-Kac theorem as the pricing equation for a continuous strip of positive coupons, namely
\Eq{
w(0,s) = \frac{1}{2} \,\Ex{0}{\int_0^t ( \eta^2[\widebar{\sigma}] - \eta^2[\underline{\sigma}] ) s^2 \partial_s^2 {\underline c} \,du } \ge0
}%
Since the relationship between the Black price w.r.t.\ the implied volatility is monotone increasing we have proved the Lemma.\end{proof}

We can proceed with the proof of the Proposition~\ref{prop:asymptotics}. We start by calculating the derivatives occurring in~\eqref{eq:locvol} from the one w.r.t.\ the time-to-maturity, we get
\begin{equation}
\label{eq: dt call}
\begin{split}
\partial_t c^{\rm BS}(t,k,\sigma(t,k)) &= \Theta^{\rm BS}(t,k,\sigma(t,k))+\mathcal{V}^{\rm BS}(t,k,\sigma(t,k)) \partial_t \sigma(t,k) \\
&= \mathcal{V}^{\rm BS}(t,k,\sigma(t,k)) \left( \frac{\sigma(t,k)}{2t} + \partial_t \sigma(t,k) \right)
\end{split}
\end{equation}
where the Black Theta is given by
\Eq{
\Theta^{\rm BS}(t,k,\sigma) := \partial_t c^{\rm BS}(t,k,\sigma) = \phi(y+\sigma \sqrt{t}) \frac{\sigma \sqrt{t}}{2} t^{-1} = \frac{\sigma}{2t} \mathcal{V}^{\rm BS}(t,k,\sigma)
}%
We go on with the first derivative w.r.t.\ the strike
\begin{equation}
\label{eq: dk call}
\begin{split}
\partial_k c^{\rm BS}(t,k,\sigma(t,k)) &= \bar{\Delta}^{\rm BS}(t,k,\sigma(t,k))+\mathcal{V}^{\rm BS}(t,k,\sigma(t,k)) \partial_k \sigma(t,k) \\
&= -\Phi(y) + \mathcal{V}^{\rm BS}(t,k,\sigma(t,k)) \partial_k \sigma(t,k)
\end{split}
\end{equation}
where the Black dual Delta is given by
\Eq{
\bar{\Delta}^{\rm BS}(t,k,\sigma) := \frac{\partial c^{\rm BS}(t,k,\sigma)}{\partial k} = - \Phi(y)
}%
and the second derivative w.r.t.\ the strike
\begin{equation}
\label{eq: dkk call}
\begin{split}
\partial_k^2 c^{\rm BS}(t,k,\sigma(t,k)) &= \bar{\Gamma}^{\rm BS}(t,k,\sigma(t,k))+2 \bar{\mathcal D}^{\rm BS}(t,k,\sigma(t,k)) \partial_k \sigma(t,k)+ \\ 
& \mathcal{V}^{\rm BS}(t,k,\sigma(t,k)) \partial_k^2 \sigma(t,k) + {\mathcal W}^{\rm BS}(t,k,\sigma) (\partial_k \sigma(t,k))^2 \\
&= \mathcal{V}^{\rm BS}(t,k,\sigma(t,k)) \bigg( \frac{1}{k^2 \sigma(t,k) t} + 2 \frac{y+\sigma(t,k) \sqrt{t}}{k \sigma(t,k) \sqrt{t}} \partial_k \sigma(t,k) + \\
& \partial_k^2 \sigma(t,k) + \frac{y(y+\sigma(t,k) \sqrt{t})}{\sigma(t,k)} (\partial_k \sigma(t,k))^2  \bigg)
\end{split}
\end{equation}
where the Black dual Gamma, dual Vanna and Volga are given by
\Eq{
\bar{\Gamma}^{\rm BS}(t,k,\sigma) := \partial_k^2 c^{\rm BS}(t,k,\sigma) = \frac{\phi(y)}{k \sigma \sqrt{t}} = \frac{\phi(y+\sigma \sqrt{t})}{k^2 \sigma \sqrt{t}} = \frac{1}{k^2 \sigma t} \mathcal{V}^{\rm BS}(t,k,\sigma)
}%
\Eq{
\bar{\mathcal D}^{\rm BS}(t,k,\sigma) := \partial_k\partial_\sigma c^{\rm BS}(t,k,\sigma) = \frac{y+\sigma \sqrt{t}}{k \sigma \sqrt{t}} \mathcal{V}^{\rm BS}(t,k,\sigma)
}%
\Eq{
{\mathcal W}^{\rm BS}(t,k,\sigma) := \partial_\sigma^2 c^{\rm BS}(t,k,\sigma) = \frac{y(y+\sigma \sqrt{t})}{\sigma} \mathcal{V}^{\rm BS}(t,k,\sigma)
}%

If we substitute~\eqref{eq: dt call},~\eqref{eq: dk call} and~\eqref{eq: dkk call} in~\eqref{eq:locvol} (dropping the dependence of $\sigma(t,k)$ on $t$ and $k$ to maintain notation simpler) we obtain the local volatility in terms of the implied volatility as given by
\begin{equation}
\eta^2(t,k)
= \frac{
    \sigma^2
  + 2 \sigma t \,\partial_t \sigma
  + 2 a \sigma t ( \Lambda(t,k,\sigma) + (1-k) \partial_k \sigma )
}
{
  1
  + 2 k \sqrt{t}(y+ \sigma \sqrt{t}) \partial_k \sigma
  + k^2 \sigma t \partial_k^2\sigma
  + k^2 t y (y+\sigma \sqrt{t}) (\partial_k \sigma)^2
}
\label{eq: rearranged exact local}
\end{equation}%
where we define
\Eq{
\Lambda(t,k,\sigma) := \frac{\Phi(y + \sigma \sqrt{t}) - \Phi(y)}{\mathcal{V}^{\rm BS}(t,k,\sigma)}
}%

Now, we can proceed as in~\cite{Berestycki2002} by defining the formal small-time-limit solution $\psi(k)$ of Equation~\eqref{eq: rearranged exact local} as given by the unique positive solution of
\begin{equation}
\psi^2(k) - \eta^2(0,k) \left( 1 - k \log k \,\partial_k \log \psi(k) \right)^2 = 0
\label{eq: formal limit local}
\end{equation}%
In calculating the above expression we have used Lemma~\ref{lemma:limit} to evaluate the small time limit of $\Lambda$ by assuming that the implied volatility limit exists. Then, we have to prove that the small-time limit of the implied volatility function is indeed given by $\psi$. This can be accomplished by following~\cite{Berestycki2002} and using Lemma~\ref{lemma:comparison} when the comparison principle is required. Here, we sketch only the main reasoning of the proof, while all technical details can be found in the cited paper. We can define the implied volatility functions
\Eq{
\widebar{\sigma}(t,k) := \psi(k) ( 1 + \kappa t )
\;,\quad
\underline{\sigma}(t,k) := \psi(k) ( 1 - \kappa t )
}%
with $\kappa>0$, which are well-defined also in $t=0$. We can calculate the corresponding local volatility functions $\eta[\underline{\sigma}](t,k)$ and $\eta[\widebar{\sigma}](t,k)$ obtained by using Equation~\eqref{eq:locvol}. Then, we can show by direct computation that for any choice of $\kappa$ we can find a time interval $[0,\delta]$ where $\eta[\underline{\sigma}] \le \eta \le \eta[\widebar{\sigma}]$. Hence, by invoking the comparison principle given by Lemma~\ref{lemma:comparison} we obtain $\underline{\sigma} \le \sigma \le \widebar{\sigma}$, and by taking the small-time limit we get that
\Eq{
\sigma(0,k) := \lim_{t\rightarrow 0} \sigma(t,k) = \psi(k)
}%

Finally, to prove the Proposition~\ref{prop:asymptotics} we can solve Equation~\eqref{eq: formal limit local} to obtain Equation~\eqref{eq:harmonic}. By taking the derivative w.r.t.\ to the strike $k$ and taking the limits for near-ATM strikes, namely $k\rightarrow 1$, we get also Equation~\eqref{eq:asymptotics}, and we complete the proof of the proposition.
\end{document}